\documentclass[a4paper]{sig-alternate}
\usepackage{amssymb}
\usepackage{bbm}
\usepackage[utf8]{inputenc}
\usepackage{algo}
\usepackage{amscd}

\def\<#1>{\langle#1\rangle}
\let\set\mathbbm

\def\lc{\operatorname{lc}}

\def\tends#1{\raisebox{2pt}{\!$\begin{CD}@>{#1}>>\end{CD}$\!}}

\def\thepage{} 

\newtheorem{theorem}{Theorem}

\newtheorem{example}{Example}
\newtheorem{rem}{Remark}
\newtheorem{conjecture}{Conjecture}
\newtheorem{@algorithm}{Algorithm}
\newenvironment{algorithm}{\begin{@algorithm}\normalfont}{\end{@algorithm}}

\def\true{\mathrm{True}}
\def\false{\mathrm{False}}
\def\field{K}

\conferenceinfo{ISSAC 2010,}{25--28 July 2010, Munich, Germany.} 
\CopyrightYear{2010} 
\crdata{...}

\usepackage[plainpages=false,pdfpagelabels,colorlinks=true,citecolor=blue,hypertexnames=false]{hyperref}  

\begin{document}

 \title{When can we detect that a P-finite sequence is positive?}


\numberofauthors{2}
\author{
 \alignauthor Manuel Kauers\titlenote{Supported by the Austrian Science Fund
   (FWF) grants P20347-N18 and Y464-N18.}\\[\smallskipamount]
      \affaddr{RISC}\\
      \affaddr{Johannes Kepler University}\\
      \affaddr{4040 Linz (Austria)}\\[\smallskipamount]
      \email{mkauers@risc.jku.at}
 \alignauthor Veronika Pillwein\titlenote{Supported by the Austrian Science Fund
   (FWF) grant W1214/DK6.}\\[\smallskipamount]
      \affaddr{RISC}\\
      \affaddr{Johannes Kepler University}\\
      \affaddr{4040 Linz (Austria)}\\[\smallskipamount]
      \email{vpillwei@risc.jku.at}
}

\maketitle

\begin{abstract}
We consider two algorithms which can be used for proving positivity of
sequences that are defined by a linear recurrence equation with 
polynomial coefficients (P-finite sequences). 
Both algorithms have in common that while they do succeed on a great
many examples, there is no guarantee for them to terminate, and they
do in fact not terminate for every input.
For some restricted classes of P-finite recurrence equations of order 
up to three we provide a priori criteria that assert the 
termination of the algorithms.
\end{abstract}




\section{Introduction}



Inequalities for special functions are a serious challenge, both from the
traditional paper-and-pencil point of view, but also (and in particular) for
computer algebra.  In contrast to the vast number of algorithms for dealing with
identities, almost no algorithms are available for inequalities.  Already for
the very restricted class of sequences satisfying linear recurrence equations
with constant coefficients (C-finite sequences), the positivity problem leads to
hard number theoretic questions to which no solutions are known today,
see~\cite{everest03,gerhold05c} and the references given there for the current
state of the struggle.

Still, inequalities are not entirely hopeless. For example, Mezzarobba and Salvy
have recently given an algorithm for effectively computing tight upper bounds
for sequences defined by linear recurrence equations with polynomial
coefficients (P-finite sequences)~\cite{mezzarobba09}.  Five years ago, Gerhold
and Kauers~\cite{gerhold05a} proposed a method applicable to inequalities
concerning quantities that satisfy recurrence equations of a very general type.
Their method consists of constructing a sequence of polynomial sufficient
conditions that would imply the non-polynomial inequality under
consideration. If one of the conditions in the sequence happens to be true
(which can be detected, e.g., with Cylindrical Algebraic Decomposition~\cite{collins75,collins91,caviness98,basu06}), the
method succeeds, otherwise it keeps on running forever. Simultaneously, the method
searches for counterexamples and it will find one and terminate for every false
inequality.


Despite its simplicity, the method has proven quite successful in
applications. Not only did it provide the first computer proofs of some
special function inequalities from the
literature~\cite{gerhold05a,gerhold06b,kauers07b,kauers07q}, but it even helped
to resolve some open
conjectures~\cite{alzer06a,kauers06i,kauers07q,pillwein08a}. At the same time,
the method remains somewhat unsatisfactory from a computational point of view,
as it is not clear on which inequalities it succeeds and on which it doesn't. It
would be interesting to have, at least for some restricted classes, some a priori
criteria telling us whether the method (or some variation of it) will succeed or
not.


Our goal in this paper is to provide such criteria for two particular proving
procedures (Algorithms~\ref{algo:1} and~\ref{algo:2} described below).
We are far from being able to give a full answer to the question posed in the title, 
but we can identify some nontrivial portions of P-finite recurrence equations of fixed order
on which termination of Algorithms \ref{algo:1} or~\ref{algo:2} is guaranteed.
For first order equations, deciding positivity is trivial.
For second order equations, we provide a result 
(Theorems~\ref{thm:2} and~\ref{thm:3}) that answers the question under a
genericity assumption.
For third order equations, we are able to identify the terminating cases of
Algorithm~\ref{algo:2} but only have partial results for Algorithm~\ref{algo:1}
supplemented by empirical evidence supporting a conjecture concerning its
terminating cases. 
An interesting aspect of our analysis is that algorithms for real quantifier 
elimination are not only used as a subroutine of Algorithms~\ref{algo:1} and~\ref{algo:2},
but they are also contributing in an essential way to the proofs of our termination
theorems. 
It is therefore possible---in principle---to extend our results to equations of
order greater than three. 
Only the increasing time and memory requirements of the computations 
have prevented us from doing so.

\section{Preliminaries}\label{sec:prelim}

A sequence $f\colon\set N\to\field:=\set R\cap\bar{\set Q}$ is called \emph{P-finite} (or
\emph{holonomic}) if there exist polynomials $p_0,\dots,p_r\in\field[x]$, not all zero,
such that
\[
  p_0(n)f(n)+p_1(n)f(n+1)+\cdots+p_r(n)f(n+r)=0
\]
for all $n\in\set N$. Such an equation is called a (P-finite) \emph{recurrence,} and $r$ is
called its order. 
If $p_r(n)\neq0$ for all $n\in\set N$, then the infinite sequence~$f$ is uniquely 
determined by the recurrence and $r$ \emph{initial values} $f(0),f(1),\dots,f(r-1)$. 
The assumption $p_r(n)\neq0$ for all $n\in\set N$ can be adopted without loss of
generality, because we can substitute $g(n)=f(n+u)$ for some $u$ larger than the
greatest integer root of~$p_r$ and then consider $g$ instead of $f$ and check
nonnegativity of the finitely many terms $f(0),f(1),\dots,f(u-1)$ by inspection. 
We will do so.
\begin{center}
 \fbox{\parbox[t]{.8\hsize}{From now on, 
     all recurrences are assumed to have a leading coefficient $p_r$ with no
     positive integer roots.}}
\end{center}

A P-finite recurrence is called \emph{balanced} if $\deg p_0=\deg p_r$ 
and $\deg p_i\leq \deg p_0$ ($i=1,\dots,r$).
The \emph{characteristic polynomial} of a balanced recurrence is defined as
\[
  \lc_y\bigl(p_0(y)+p_1(y)x+p_2(y)x^2+\cdots+p_r(y)x^r\bigr)\in\field[x].
\]
Its roots $\lambda_1,\dots,\lambda_r\in\set C$ are called the \emph{eigenvalues}
of the recurrence. (The $\lambda_i$ are not necessarily distinct.)
Note that for a balanced recurrence, the characteristic polynomial has
always degree~$r$ and it has never $0$ as a root.

An eigenvalue $\lambda_i$ is called \emph{dominant} if $|\lambda_j|\leq|\lambda_i|$
for all $j=1,\dots,r$. Dominant eigenvalues govern the asymptotics of the 
sequences defined by the recurrence~\cite{wimp85,flajolet09}. 
If there is a unique dominant eigenvalue~$\lambda_i$, then for we will usually have
\[
  f(n)\sim c(n)\lambda_i^n
  \quad(n\to\infty)
\]
where $c$ is of subexponential growth in the sense that
\[
  \frac{c(n+1)}{c(n)}\tends{n\to\infty}1.
\]
There may be choices of initial values for which $c(n)=0$ for all~$n$
so that the asymptotics of $f$ is not affected by $\lambda_i$ but by the
next smaller eigenvalue(s). 
Whether this is the case or not can be hard to verify formally, but is usually
easy to verify empirically. 
Some of our termination results apply only to this \emph{generic} situation where 
initial values are chosen such as to actually exhibit the asymptotic behavior 
predicted by the dominant eigenvalue. 

Finally, if the dominant eigenvalue $\lambda_i$ is not real and positive, 
then it is clear that $f$ will be ultimately oscillating, and so $f(n)\geq0$
cannot possibly be true for all~$n$. This case can be sorted out trivially
beforehand, and we may therefore assume that the unique dominant eigenvalue is
real and positive. In this case, the substitution $g(n)=f(n)/\lambda_i^n$
turns the recurrence
\[
  p_0(n)f(n)+p_1(n)f(n+1)+\cdots+p_r(n)f(n+r)=0
\]
into
\[
  p_0(n)g(n)+p_1(n)\lambda_i g(n+1)+\cdots+p_r(n)\lambda_i^rg(n+r)=0
\]
whose dominant eigenvalue is~$1$. As $g(n)\geq0\iff f(n)\geq0$, it suffices
to consider this case.

\section{Induction Based Proving\hskip0ptplus1fill\break Procedures}

\subsection{The Original Version}\label{sec:gk}

The approach of~\cite{gerhold05a} is as follows. Suppose that $f\colon\set N\to\field$
is defined by a recurrence
\[
  p_0(n)f(n)+p_1(n)f(n+1)+\cdots+p_r(n)f(n+r)=0
\]
and initial values $f(0)=f_0,f(1)=f_1,\dots,f(r-1)=f_{r-1}$.
We seek to prove $f(n)\geq0$ for all $n\in\set N$ by induction:
\[
  f(n)\geq0\land\cdots\land f(n+r-1)\geq0\Longrightarrow
  f(n+r)\geq0.
\]
Because of the recurrence, this is equivalent to 
\begin{alignat*}1
  &f(n)\geq0\land\cdots\land f(n+r-1)\geq0\\
  &\quad\Longrightarrow
  -\frac{p_0(n)}{p_r(n)}f(n)
  -\cdots
  -\frac{p_{r-1}(n)}{p_r(n)}f(n+r-1)\geq0
\end{alignat*}
For this to be true for all $n\in\set N$, it is sufficient that the
\emph{induction step formula}
\begin{alignat*}1
  &\forall\ y_0,y_1,\dots,y_{r-1}\in\set R
  \ \forall\ x\in\set R:\\
  &\quad \bigl(x\geq0 \land y_0\geq0\land\cdots\land y_{r-1}\geq0\bigr)\\
  &\qquad \Longrightarrow 
  -\frac{p_0(x)}{p_r(x)}y_0
  -\cdots
  -\frac{p_{r-1}(x)}{p_r(x)}y_{r-1}\geq0
\end{alignat*}
is true, and this can be decided by a quantifier elimination algorithm. 
If it is true, the induction step is established and $f$ is nonnegative
everywhere if and only if it is nonnegative for $n=0,\dots,r-1$, which can
be checked.

In the unlucky case when the induction step formula is false, there is
no immediate conclusion about $f$ that could be drawn. In this case, 
\emph{refined induction step formulas} 
\[
  f(n)\geq0\land\cdots\land f(n+\varrho-1)\geq0\Longrightarrow
  f(n+\varrho)\geq0
\]
for $\varrho>r$ are constructed. Using the recurrence, each
term $f(n+i)$ can be rewritten as a linear combination of $f(n)$, \dots, $f(n+r-1)$
with rational function coefficients, and using this rewriting, the 
refined induction step formula takes the form
\begin{alignat*}1
\Phi(\varrho):={}
  &\forall\ y_0,y_1,\dots,y_{r-1}\in\set R
  \ \forall\ x\in\set R:\\
  &\quad \bigl(x\geq0 \land y_0\geq0\land\cdots\land y_{r-1}\geq0\\
  &\qquad \land q_{r,0}(x)y_0+\cdots+q_{r,r-1}(x)y_{r-1}\geq0\\
  &\qquad \land q_{r+1,0}(x)y_0+\cdots+q_{r+1,r-1}(x)y_{r-1}\geq0\\[-4pt]
  &\qquad\,\ \vdots \\
  &\qquad \land q_{\varrho-1,0}(x)y_0+\cdots+q_{\varrho-1,r-1}(x)y_{r-1}\geq0\bigr)\\
  &\qquad\quad \Longrightarrow 
     q_{\varrho,0}(x)y_0+\cdots+q_{\varrho,r-1}(x)y_{r-1}\geq0,
\end{alignat*}
where the $q_{i,j}$ are some rational functions. 

The full method then reads as follows.

\begin{algorithm}\label{algo:gerhold-kauers}\label{algo:1} \null \hfill\\
\textit{Input:} A P-finite recurrence of order~$r$ and a vector of initial values defining a 
sequence $f\colon\set N\to\field$.\\
\textit{Output:} $\true$ if $f(n)\geq0$ for all $n\in\set N$,
$\false$ if $f(n)<0$ for some $n\in\set N$, 
possibly no output at all.

\kern-\smallskipamount

\begin{algo}%
|for| n = 0 |to| r-1 |do|
^^I |if| f(n) <0 |then| |return| \false
|for| n = r,r+1,r+2,r+3,\dots |do|\label{algo:1:3}
^^I |if| \Phi(n) |then| |return| \true
^^I |if| f(n)< 0 |then| |return| \false\label{algo:1:5} %
\end{algo}
\end{algorithm}

\begin{example}
  Let $f\colon\set N\to\field$ be defined by
  \begin{alignat*}1
    &(2n+13)f(n+3)-(5n+22)f(n+2)\\
    &\quad{}+(3n+20)f(n+1)-(2n+7)f(n)=0,\\
    &f(0)=f(1)=f(2)=1.
  \end{alignat*}
  We use Algorithm~\ref{algo:1} to show that $f(n)\geq0$ for all $n\in\set N$.

  Since $f(0),f(1),f(2)\geq0$, we enter the loop in line~\ref{algo:1:3}.
  For $n=3$, we have
  \begin{alignat*}1
   \Phi(n)&=\forall\ y_0,y_1,y_2\ \forall\ x\in\set R:\\
   &\quad\bigl(x\geq0\land y_0\geq0\land y_1\geq0\land y_2\geq0\bigr)\\
   &\qquad\Longrightarrow\tfrac{2x+7}{2x+13}y_0-\tfrac{3x+20}{2x+13}y_1+\tfrac{5x+22}{2x+13}y_2\geq0.
  \end{alignat*}
  This is false, but since $f(3)=9/13>0$ (checked in line~\ref{algo:1:5}), we continue.

  The formula $\Phi(4)$ is too lengthy to be reproduced here explicitly, and it
  is also false.  Yet $f(4)=61/195\geq0$, so we proceed to consider the even
  lengthier formula~$\Phi(5)$, which turns out to be true. At this point the
  algorithm terminates with output $\true$.
\end{example}

\subsection{A Variation}

In cases where Algorithm~\ref{algo:gerhold-kauers} does not terminate, it is 
sometimes possible to prove inductively the stronger statement that $f(n)$ is increasing, 
viz. that $f(n+1)\geq f(n)$ for all $n\geq0$. 
While this is obviously a sufficient condition for $f(n)\geq0$ for all~$n$,
there are of course sequences $f$ which are non-negative but not increasing.
For such cases, a good strategy is to prove that $\mu^{-n} f(n)$ is increasing,
for some suitably chosen constant $\mu>0$. 
The choice of $\mu$ is critical in two respects: it must be small enough to 
assure that $\mu^{-n} f(n)$ actually is increasing, and it must be big enough
to allow for an inductive proof. 
The following algorithm proves positivity of a P-finite sequence~$f$ by 
searching for a $\mu$ that meets both criteria. 

\begin{algorithm}\label{algo:2} \null \hfill\\
\textit{Input:} A P-finite recurrence of order~$r$ and a vector of initial values defining a 
sequence $f\colon\set N\to\field$.\\
\textit{Output:} $\true$ if $f(n)\geq0$ for all $n\in\set N$,
$\false$ if $f(n)<0$ for some $n\in\set N$, 
possibly no output at all.

\kern-\smallskipamount

\begin{algo}%
\parbox[t]{.9\hsize}{Determine a quantifier free formula $\Phi(\xi,\mu)$ equivalent to %
\begin{alignat*}1 %
  &\forall\ y_0,\dots,y_{r-1}\ \forall\ x\geq \xi: \\ %
  &\quad \Bigl(y_0\geq0 \land  y_1\geq \mu y_0 \land \cdots \land y_{r-1}\geq\mu y_{r-2}\Bigr)\\ %
  &\qquad\Longrightarrow -\frac{p_0(x)}{p_r(x)}y_0 -\cdots -\frac{p_{r-1}(x)}{p_r(x)}y_{r-1} \geq \mu y_{r-1} %
\end{alignat*}}\label{algo:3:4}
^^I|for| n=0,1,2,3,\dots |do|\label{algo:3:5}
^^I^^I|if| f(n)<0 |then|
^^I^^I^^I |return| \false
^^I^^I\!\!
\parbox[t]{.775\hsize}{$|else| |if| \exists\ \mu\geq0:\Phi(n,\mu)\land f(n+1)\geq\mu f(n)$\\ %
$\null\qquad{}\land\cdots\land f(n+r-1)\geq\mu f(n+r-2) |then|$}\label{algo:3:8}
^^I^^I^^I |return| \true %
\end{algo}
\end{algorithm}

\begin{theorem}
  Algorithm~\ref{algo:2} is correct.
\end{theorem}
\begin{proof}
  Correctness is obvious whenever the algorithm returns~$\false$, 
  because this happens only when an explicit point $n$ with $f(n)<0$ has been found. 
  Suppose now that the algorithm returns~$\true$ at the $n$th iteration of the for loop.
  Then $f(k)\geq0$ for $k=0,\dots,n$, otherwise the algorithm would have terminated
  in an earlier iteration with output $\false$.
  The condition in line~\ref{algo:3:8} inductively implies 
  \[
    \exists\ \mu\geq0 \ \forall\ k \geq n : f(k+1)\geq \mu f(k).
  \]
  Since $\mu\geq0$ and $f(n)\geq0$, this inductively implies $f(k)\geq0$ also for all $k>n$.
\end{proof}

\begin{example}
  Let $f\colon\set N\to\field$ be defined by
  \begin{alignat*}1
    &(n+3)f(n+3)-(5n+13)f(n+2)\\
    &\quad{}+(5n+12)f(n+1)-(n+2)f(n)=0,\\
    &f(0)=1,\ f(1)=1/4,\ f(2)=1/10.
  \end{alignat*}
  Algorithm~\ref{algo:gerhold-kauers} does not seem to terminate for this sequence.
  But Algorithm~\ref{algo:2} succeeds.

  Step 1 produces the quantifier free formula
  \[
    \xi \geq 0\land \frac{5-\sqrt{5}}{2}\leq \mu \leq \frac{\sqrt{5 \xi ^2+22
    \xi +25}+5 \xi +13}{2 (\xi +3)},
  \]
  which we denote $\Phi(\xi,\mu)$.
  In the iteration of the for loop, we get:

  For $n=0$, since $f(0)=1\geq0$, we check whether 
  \[
    0\geq0\land \frac{5-\sqrt5}2\leq\mu\leq 3
    \land \frac14\geq\mu 
    \land \frac1{10}\geq\frac\mu4
  \]
  is satisfiable. As it is not, we proceed.

  Also $n=1$ and $n=2$, we have $f(n)\geq0$ but there is no $\mu\geq0$ with 
  \[
   \Phi(n,\mu)\land f(n+1)\geq\mu f(n)\land f(n+2)\geq\mu f(n+1).
  \]
  Then for $n=3$, since $f(3)\geq0$, we check whether
  \[
    3\geq0\land
    \frac{5-\sqrt{5}}2\leq\mu\leq\frac{28+2\sqrt{34}}{12}
    \land\frac{17}{80}\geq\frac{\mu}{10}
    \land\frac{247}{400}\geq\frac{17\mu}{80}
  \]
  is satisfiable. As it is satisfiable (e.g., by~$\mu=2$)
  the algorithm terminates with output $\true$.
\end{example}


The two strategies employed in Algorithms~\ref{algo:1} and~\ref{algo:2} in the case where
a direct proof of the induction step formula fails (prolonging the induction hypothesis in
Algorithm~\ref{algo:1} versus multiplying with a positivity preserving exponential in
Algorithm~\ref{algo:2}) are independent of each other.  It is possible to merge both
strategies into a single strategy that simultaneously prolongs the induction hypothesis
and inserts a positivity preserving exponential.  An algorithm based on such a combined
strategy is easily written down, but turns out to be computationally quite expensive on
examples. It would be interesting to carry out the termination analysis given below for
the combined algorithm, but the quantifier elimination problems arising in this analysis
seem currently too hard to be carried out for the combined algorithm.







\section{Terminating Cases}

Both algorithms given in the previous section may fail to terminate. 
Our goal now is to identify classes of P-finite recurrence equations 
for which termination can be guaranteed a priori. 

\subsection{Order One}

This case is rather simple and included here merely for the sake of completeness. 
If $f\colon\set N\to\field$ satisfies
\[
  p_0(n)f(n)+p_1(n)f(n+1)=0,
\]
then $f(n)\geq0$ for all $n\in\set N$ if and only if $f(0)\geq0$ and
$-p_0(n)/p_1(n)\geq0$ for all $n\in\set N$. 
Since sign changes of $-p_0(n)/p_1(n)$ can occur only at the real roots of
$p_0$ or~$p_1$, the only thing we need to do is to find an upper bound
$n_0\in\set R$ for the real roots (this can be done), and then check
whether $-p_0(n)/p_1(n)\geq0$ for $n=0,1,2,\dots,n_0+1$. 

\begin{example}
  Consider $f\colon\set N\to\field$ defined via
  \[
    (3n-16)f(n)-(3n-17)f(n+1)=0,\quad f(0)=1.
  \]
  The roots of $p_0,p_1$ are $16/3$ and $17/3$, respectively, 
  and they are both less than $n_0=6$, for instance. 
  Therefore, since $f(n)\geq0$ for $n=0,\dots,6$, we can conclude
  that $f(n)\geq0$ for all~$n\in\set N$.
\end{example}

\subsection{Order Two}

We now turn to sequences $f\colon\set N\to\field$ which are defined by a balanced P-finite 
recurrence of second order,
\[
  p_2(n)f(n+2)-p_1(n)f(n+1)-p_0(n)f(n)=0.
\]
We assume (without loss of generality) that $1$ is a dominant eigenvalue of this recurrence and 
let $u\in\field$ with $|u|<1$ be such that 
\[
  (x-1)(x-u)=x^2-(u+1)x-(-u)
\]
is the characteristic polynomial of the recurrence. 
The question is whether Algorithm~\ref{algo:1} and Algorithm~\ref{algo:2} succeed in
proving that $f(n)\geq0$ for all~$n$. (If this is actually the case; if it is not,
then both algorithms will obviously succeed in finding a counterexample.)
We will show that termination of Algorithm~\ref{algo:1} depends on the sign of~$u$
whereas Algorithm~\ref{algo:2} (generically) terminates for all~$u$.

\begin{theorem}\label{thm:2}
  If $u\in(-1,0)$, then Algorithm~\ref{algo:1} terminates.
\end{theorem}
\begin{proof}
  Rewrite the recurrence in the form
  \[
    f(n+2)=\frac{p_1(n)}{p_2(n)}f(n+1)+\frac{p_0(n)}{p_2(n)}f(n).
  \]
  Since the characteristic polynomial is 
  \[
    (x-1)(x-u)=x^2-(u+1)x-(-u),
  \]
  we have
  \begin{alignat*}3
    \frac{p_1(n)}{p_2(n)}\tends{n\to\infty} u+1>0
    \text{ \ and \ }
    \frac{p_0(n)}{p_2(n)}\tends{n\to\infty} {-u} >0.
  \end{alignat*}
  Therefore,
  \[
    \exists\ n_0\in\set N\ \forall\ n\geq n_0:
    \frac{p_1(n)}{p_2(n)}>0
    \land
    \frac{p_0(n)}{p_2(n)}>0,
  \]
  where we may safely regard $n$ as ranging not only over the integers
  but over all reals except the roots of~$p_2$. 
  We show that Algorithm~\ref{algo:1} terminates after at most $n_0$
  iterations. 
  The previous formula implies
  \begin{alignat*}1
    &\forall\ y_0,y_1\ \forall\ n\geq n_0:
    \bigl(y_0\geq0\land y_1\geq 0\bigr)\\
    &\quad
    \Longrightarrow \frac{p_0(n)}{p_2(n)}y_0 + \frac{p_1(n)}{p_2(n)}y_1\geq0.
  \end{alignat*}
  Substituting $n\mapsto n-n_0$ leads to
  \begin{alignat*}1
    &\forall\ y_0,y_1\ \forall\ n\geq 0:
    \bigl(y_0\geq0\land y_1\geq 0\bigr)\\
    &\quad\Longrightarrow \frac{p_0(n+n_0)}{p_2(n+n_0)}y_0 + \frac{p_1(n+n_0)}{p_2(n+n_0)}y_1\geq0.
  \end{alignat*}
  As the variables $y_0,y_1$ range over all reals, the latter formula will remain
  true if we apply a substitution
  \begin{alignat*}1
     &y_0\mapsto r_0(n)y_0+r_1(n)y_1, \\
     &y_1\mapsto r_2(n)y_0+r_3(n)y_1,
  \end{alignat*}
  where $r_0,\dots,r_3$ are some rational functions in~$n$. 
  This gives
  \begin{alignat*}1
    &\forall\ y_0,y_1\ \forall\ n\geq 0:
    \bigl(r_0(n)y_0+r_1(n)y_1\geq0\\
    &\hphantom{\forall\ y_0,y_1\ \forall\ n\geq 0:\bigl(}\quad \land r_2(n)y_0+r_3(n)y_1\geq0\bigr)\\
    &\quad\Longrightarrow \frac{p_0(n+n_0)r_0(n)+p_1(n+n_0)r_2(n)}{p_2(n+n_0)}y_0\\
    &\quad\hphantom{{}\Longrightarrow{}}\quad{+}\frac{p_0(n+n_0)r_1(n)+p_1(n+n_0)r_3(n)}{p_2(n+n_0)}y_1\geq0.
  \end{alignat*}
  We are free to further modify this formula, without harming its truth, 
  by imposing arbitrary additional conditions on the left hand side of the implication.

  By choosing $r_0,r_1,r_2,r_3$ such that 
  \begin{alignat*}1
    f(n+n_0)&=r_0(n)f(n)+r_1(n)f(n+1)\\
    f(n+n_0+1)&=r_2(n)f(n)+r_3(n)f(n+1)
  \end{alignat*}
  and adding constraints $q_{i,0}(n)y_0+q_{i,1}(n)y_1\geq0$ encoding $f(n+i)\geq0$ ($i=0,\dots,n_0-1$)
  on the hypothesis part, we obtain precisely the formula $\Phi(n_0)$ as defined in Section~\ref{sec:gk}
  and used in Algorithm~\ref{algo:1}.
  As the formula is true, the algorithm terminates in the $n_0$th iteration (or earlier),
  as we wanted to show.
\end{proof}

\begin{rem}
  Algorithm~\ref{algo:1} fails to terminate for positive~$u$.
  To see this, consider the C-finite recurrence
  \[
    f(n+2)-(u+1)f(n+1)+u f(n)=0
  \]
  for some $u\in(0,1)$. If Algorithm~\ref{algo:1} applied to
  this recurrence terminated in the $n_0$th iteration, for some
  $n_0\geq0$, then the truth of $\Phi(n_0)$ implies that \textbf{no}
  solution $f\colon\set N\to\field$ of the recurrence can have
  $n_0$ consecutive nonnegative terms followed by a negative term.
  (So that, if $n_0$ consecutive terms are found nonnegative, all
  subsequent terms must be nonnegative as well.)

  To see that no such $n_0$ can exist for the C-finite recurrence
  above, it is sufficient to construct for every $n_0\geq0$ a 
  solution which contains a run of exactly $n_0$ nonnegative
  terms. The general solution of the recurrence
  is $c_0 + c_1 u^n$ for some constants $c_0,c_1\in\field$.
  It is easily checked that the choice
  $c_0=-1$, $c_1=u^{-n_0+1}$ has the desired property.

  The argument extends, at least for generic initial values, to
  P-finite balanced recurrence equations, using the fact that
  the recurrence admits two solutions
  $f_1,f_2\colon\set N\to\field$ with $f_1(n+1)/f_1(n)\tends{n\to\infty}1$ and 
  $f_2(n+1)/f_2(n)\tends{n\to\infty}u$. 
\end{rem}

\begin{theorem}\label{thm:3}
  If $u\in(-1,1)\setminus\{0\}$, then Algorithm~\ref{algo:2} terminates for 
  generic initial values.
\end{theorem}

\begin{proof}
  Consider the set $D_3\subseteq\set R^3$ consisting of all points $(c_0,c_1,\mu)$
  satisfying
  \[
    0<\mu <1\land \mu < c_1 < 2 \land \mu(\mu-c_1)<c_0<1
  \]
  and the set
  \[
    D_2:=\{\,(c_0,c_1)\in\set R^2: 0<c_1<2\land -\tfrac14c_1^2<c_0<1\,\}.
  \]
  It can be shown by CAD computations that 
  \begin{alignat}1
    \forall\ (c_0,c_1)\in D_2\ 
         \exists\ \mu\in(0,1):(c_0,c_1,\mu)\in D_3\label{proof:2}
  \end{alignat}
  and that
  \begin{alignat}1
    &\forall\ (c_0,c_1,\mu)\in D_3\
    \forall\ y_0,y_1\in\set R:\notag\\
    &\quad\bigl(y_0\geq 0 \land y_1\geq \mu y_0\bigr)
    \Longrightarrow c_0 y_0+c_1y_1\geq\mu y_1.\label{proof:1}
  \end{alignat}
  Since $\tfrac14(u+1)^2>u$ for all $u\in(-1,1)$, 
  the set $D_2$ contains in particular the point~$(-u,u+1)$ where $u$ is from the
  statement of the theorem. 
  Because of~\eqref{proof:2}, there exists $\mu\in(0,1)$ with $(-u,u+1,\mu)\in D_3$.
  Since $D_3$ is open, there exists $\varepsilon>0$ such that
  \[
    U := (-u-\varepsilon,-u+\varepsilon)
      \times(u+1-\varepsilon,u+1+\varepsilon)\times\{\mu\}\subseteq D_3.
  \]
  Because of
  \[
    \frac{p_0(n)}{p_2(n)}\tends{n\to\infty} {-u}\quad\text{and}\quad
    \frac{p_1(n)}{p_2(n)}\tends{n\to\infty} u+1,
  \]
  there exists $\xi\in\set N$ such that 
  \[
    \Bigl(\,\frac{p_0(n)}{p_2(n)},\frac{p_1(n)}{p_2(n)},\mu\Bigr)\in U\subseteq D_3
  \]
  for all $n\geq\xi$. 
  Together with~\eqref{proof:1}, this implies
  \begin{alignat*}1
   &
   \exists\ \mu\in(0,1)\
   \exists\ \xi\in\set N\ 
   \forall\ n\geq \xi\    
   \forall\ y_0,y_1\in\set R:\\
   &\quad\bigl(y_0\geq0\land y_1\geq\mu y_0\bigr)
   \Longrightarrow \frac{p_0(n)}{p_2(n)}y_0+\frac{p_0(n)}{p_2(n)}y_1\geq \mu y_1.
  \end{alignat*}
  Therefore, the set
  \[
    C:=\{\,(\xi,\mu)\in(0,\infty)\times(0,1):\Phi(\xi,\mu)\,\}
  \]
  with $\Phi(\xi,\mu)$ as used in Algorithm~\ref{algo:2} is not empty.

  Fix some point $(\xi,\mu)\in C$.
  Then it is immediate by the defining formula 
  that also $(\xi',\mu)\in C$ for every $\xi'>\xi$, so
  \[
    (\xi,\infty)\times\{\mu\}\subseteq C.
  \]
  Let now $f\colon\set N\to\field$ be the sequence to which Algorithm~\ref{algo:2} is
  applied.  Then, because the eigenvalues of its defining recurrence are $1$ and $u$ and
  we have $|u|<1$ and we assume generic initial values, we have
  \[
    \frac{f(n+1)}{f(n)}\tends{n\to\infty}1.
  \]
  Since $\mu<1$, this implies the existence of an index $m\in\set N$ such that $f(n)\geq0$
  for all $n\geq m$ and $f(n+1)/f(n)\geq\mu$ for all $n\geq m$, so that we get
  $f(n+1)\geq\mu f(n)$ for all $n\geq m$.

  It follows that the algorithm terminates no later than at iteration
  $\max(m,\xi)$.
\end{proof}

\begin{rem}
  The defining inequalities of $D_3$ in the preceding proof were found by
  quantifier elimination applied to formula~\eqref{proof:1} with the first
  quantifier dropped. This computation as well as the CAD computations 
  referred to in the proof were performed with Mathematica's built-in 
  implementation of CAD~\cite{strzebonski00,strzebonski06}. 
  The computation time is negligible for all of them and we are sure that
  other implementations~\cite[etc.]{dolzmann97,brown03} would have no problem with them either.
\end{rem}

\begin{example}
  The restriction to generic initial values in Theorem~\ref{thm:3} is
  essential: let $f\colon\set N\to\field$ be defined via
  \begin{alignat*}1
    &(n+3)^2f(n+2) -\tfrac12(n+2)(3n+11)f(n+1)\\
    &\qquad{}+\tfrac12(n+4)(n+1)f(n)=0
  \end{alignat*}
  and $f(0)=1, f(1)=1/4$. Then we have $f(n)=2^{-n}/(n+1)$ for all $n\in\set N$
  and so in particular $f(n)\geq0$ ($n\in\set N$). 

  Algorithm~\ref{algo:2} finds
  \[
    \Phi(\xi,\mu)\equiv 
    \tfrac12\leq\mu\leq 1\land
    \xi\geq\frac{17-12\mu-\sqrt{25-16\mu}}{2\mu-3}
  \]
  in Step~\ref{algo:3:4} and continues by searching for an index $n$
  with $f(n+1)\geq \tfrac12 f(n)$. 
  As no such index exist, the search continues forever.

  The general solution of the defining recurrence for $f$ is
  \[
    c_0+c_1\frac{2^{-n}}{n+1}
  \]
  and we will have $c_0\neq0$ for a generic choice of initial values.
  In these cases, the solution converges to $c_0$ and therefore
  eventually reaches an index $n_0$ with a term that is greater than
  half its predecessor.
\end{example}



\subsection{Order Three}

Consider now sequences $f\colon\set N\to\field$ defined by a balanced P-finite
recurrence of third order,
\[
  p_3(n)f(n+3)-p_2(n)f(n+2)-p_1(n)f(n+1)-p_0(n)f(n)=0.
\]
Again, we assume without loss of generality that $1$ is a dominant eigenvalue and we let $u,v\in\field$
be such that 
\[
  (x-1)(x^2+ux+v)=x^3-(1-u)x^2-(u-v)x-v
\]
is the characteristic polynomial of the recurrence under consideration. 
The condition that the two roots of the quadratic factor belong to the 
interior of the complex unit disc translates into the condition
\[
  |u|-1<v<1
\]
for the coefficients of the polynomial. The points $(u,v)\in\field^2$ satisfying
this condition form the interior of the triangle with corners at $(-2,1),(2,1),(0,-1)$:

\centerline{\epsfig{file=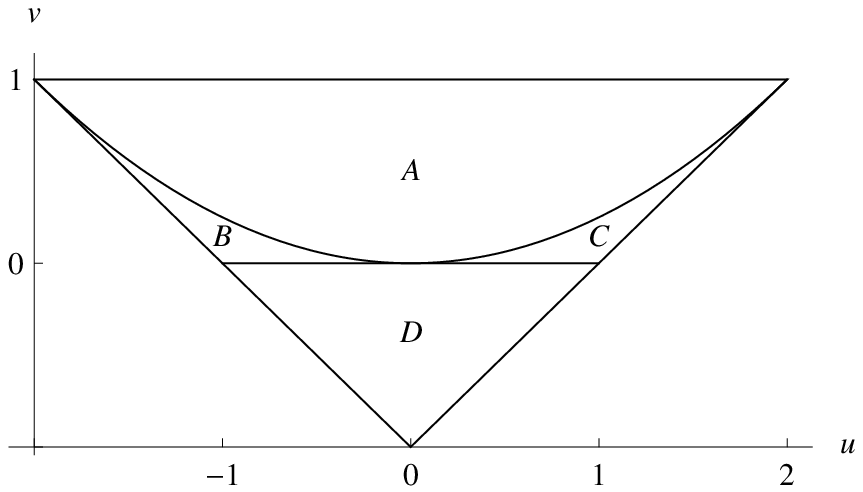,width=.8\hsize}}

Just for the sake of orientation: the polynomial $x^2+ux+v$ has two complex conjugate
roots in region~$A$, two positive real roots in region~$B$, two negative real roots
in region~$C$, and a positive as well as a negative root in region~$D$.

We want to identify regions of the triangle corresponding to recurrence
equations on which Algorithms~\ref{algo:1} and~\ref{algo:2} terminate. 
Only for Algorithm~\ref{algo:2} we have a satisfactory result, so let
us consider this case first.

\begin{theorem}\label{thm:5}
  If $|u|-1<v<1$ and $4v<(u+1)^2$ and $u<1$, then Algorithm~\ref{algo:2} terminates
  for generic initial values.
\end{theorem}

\begin{proof}
Consider the set $D_4\subseteq\set R^4$ consisting of all points $(c_0,c_1,c_2,\mu)$
satisfying 
\begin{alignat*}1
  &0<\mu <1\land \mu<c_2 \land \mu(\mu-c_2)<c_1
  \land \mu^3-c_2 \mu^2-c_1\mu<c_0
\end{alignat*}
and the set $D_3\subseteq\set R^3$ consisting of all points $(c_0,c_1,c_2)$ with
\begin{alignat*}1
  &\bigl(0<c_2<2\land -\tfrac14c_2^2<c_1<\min(3-2 c_2,c_2^2)\\
  &\qquad{}\land 2 c_2^3+9 c_1 c_2+27 c_0+2 (c_2^2+3 c_1)^{3/2}>0\bigr)\\
  &\quad \lor \bigl(0<c_2<1\land c_1\geq c_2^2\land c_0+c_1 c_2>0\bigr).
\end{alignat*}
The following facts can be verified by CAD:
\begin{alignat*}1
  \bullet\ &\forall\ (c_0,c_1,c_2)\in D_3\ 
   \exists\ \mu\in(0,1):(c_0,c_1,c_2,\mu)\in D_4\\
  \bullet\ &\forall\ (c_0,c_1,c_2,\mu)\in D_4\
   \forall\ y_0,y_1,y_2\in\set R:\notag\\
  &\quad\bigl(y_0\geq0\land y_1\geq\mu y_0\land y_2\geq\mu y_1\bigr)\notag\\
  &\qquad\Longrightarrow c_0y_0+c_1y_1+c_2y_2\geq\mu y_2\\
  \bullet\ &\forall\ u,v\in\set R:\bigl(|u|-1<v<1\land 4v<(u+1)^2 \land u<1\bigr)\notag\\
  &\qquad\Longrightarrow (v,u-v,1-u)\in D_3.
\end{alignat*}
Consequently, for $u,v$ from the statement of the theorem
there exists $\mu\in(0,1)$ such that $(v,u-v,1-u,\mu)\in D_4$.
Since $D_4$ is open, there exists $\varepsilon>0$ such that
\begin{alignat*}1
  U := {}&(v-\varepsilon,v+\varepsilon)
   \times (u-v-\varepsilon,u-v+\varepsilon)\\
   &\quad{}
   \times (1-u-\varepsilon,1-u+\varepsilon)
   \times\{\mu\}\subseteq D_4.
\end{alignat*}
Using
\[
  \frac{p_0(n)}{p_3(n)}\stackrel{n\to\infty}{\hbox to 2.5em{\rightarrowfill}} v, \quad
  \frac{p_1(n)}{p_3(n)}\stackrel{n\to\infty}{\hbox to 2.5em{\rightarrowfill}} u-v, \quad
  \frac{p_2(n)}{p_3(n)}\stackrel{n\to\infty}{\hbox to 2.5em{\rightarrowfill}} 1-u, 
\]
the rest of the proof is fully analogous to the proof of Theorem~\ref{thm:3}.
\end{proof}

The set of points $(u,v)$ for which Theorem~\ref{thm:5} asserts termination of
Algorithm~\ref{algo:2} is the shaded area in the figure below.

\centerline{\epsfig{file=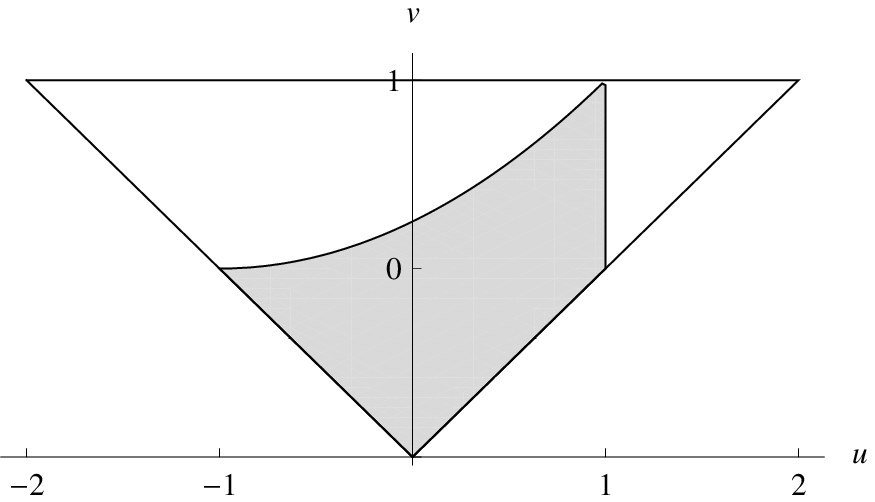,width=.8\hsize}}

The truth of the formula
\begin{alignat*}1
  &\forall\ u\in(-1,1)\ \forall\ v\in(|u|-1,1):\\
  &\quad\bigl[\exists\ \mu>0\ \forall\ y_0,y_1,y_2:\\
  &\qquad\bigl(y_0\geq0\land y_1\geq\mu y_0\land y_2\geq\mu y_1\bigr)\\
  &\qquad\Longrightarrow v y_0+(u-v)y_1+(1-u)y_2\geq\mu y_2\bigr]\\
  &\quad\Longrightarrow u < 1 \land 4v<(u+1)^2
\end{alignat*}
(as confirmed, once again, by a CAD computation) asserts that Theorem~\ref{thm:5} is sharp.

We are not able to provide a sharp result for the terminating region of Algorithm~\ref{algo:1}.
If we proceed to reason as in the proof of Theorem~\ref{thm:2}, we obtain termination for
$(u,v)$ restricted to the (open) triangle with vertices $(0,0)$, $(1,0)$,~$(1,1)$, essentially 
because of
\begin{alignat*}1
 &\forall\ u\in(0,1)\ \forall\ v\in(0,u)\
 \forall\ y_0,y_1,y_2:\\
 &\quad\bigl(y_0\geq0\land y_1\geq0\land y_2\geq0\bigr)\\
 &\quad\Longrightarrow v y_0+(u-v)y_1+(1-u)y_2\geq 0
\end{alignat*}
and the convergences
\[
  \frac{p_0(n)}{p_3(n)}\stackrel{n\to\infty}{\hbox to 2.5em{\rightarrowfill}} v, \quad
  \frac{p_1(n)}{p_3(n)}\stackrel{n\to\infty}{\hbox to 2.5em{\rightarrowfill}} u-v, \quad
  \frac{p_2(n)}{p_3(n)}\stackrel{n\to\infty}{\hbox to 2.5em{\rightarrowfill}} 1-u.
\]
But this is not the entire terminating region. 
A larger portion of the terminating region can be identified by starting out with a
formula corresponding to an induction hypothesis of length four. 
As the formula
\begin{alignat*}1
 &\forall\ y_0,y_1,y_2:\bigl(y_0\geq0\land y_1\geq0\land y_2\geq0\\
 &\qquad\quad\land v y_0+(u-v)y_1+(1-u)y_2\geq 0\bigr)\\
 &\ \Longrightarrow (1-u)v y_0 + u (1-u+v) y_1 + (1-u+u^2-v)y_2\geq0
\end{alignat*}
is true for all $(u,v)$ with
\begin{alignat*}1
  &u < 1\land v>0\land 1-u+u^2-v > 0\\
  &\quad\land\bigl(u>0\lor u^2-v-uv+v^2<0\bigr),
\end{alignat*}
and as we have
\begin{alignat*}1
  \frac{p_0(n)p_2(n+1)}{p_3(n)p_3(n+1)}&\tends{n\to\infty}(1-u)v,\\
  \frac{p_1(n)p_2(n+1)+p_0(n+1)p_3(n)}{p_3(n)p_3(n+1)}&\tends{n\to\infty}u(1-u+v),\\
  \frac{p_2(n)p_2(n+1)+p_1(n+1)p_3(n)}{p_3(n)p_3(n+1)}&\tends{n\to\infty}1-u+u^2-v,
\end{alignat*}
Algorithm~\ref{algo:1} also terminates for all $(u,v)$ satisfying the conditions
stated above. 

Starting out with a formula corresponding to a hypothesis of length five leads
to a portion of the termination region whose description can be computed in a reasonable
amount of time, but which is already too big to be reproduced here.
For longer induction hypotheses, the computational effort for doing quantifier 
elimination becomes prohibitive.
But it is still possible to determine experimentally the regions obtained by
taking a particular length~$\varrho$ of the induction hypothesis taken as the starting
point of the termination proof. 
The empiric results for induction hypotheses of length up to 10 are as follows
(the numbers indicate the length of the induction hypothesis):

\centerline{\epsfig{file=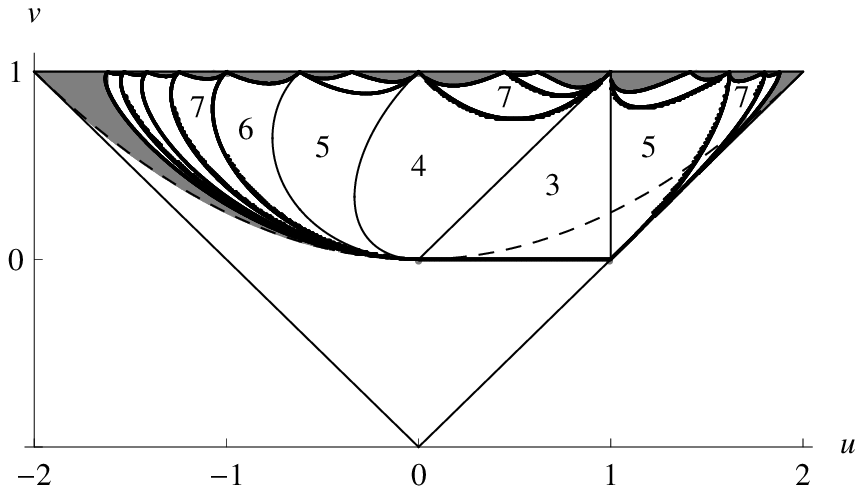,width=.8\hsize}}

The picture suggests the following characterization for the full region of termination.

\begin{conjecture}
  If $|u|-1<v<1$ and
  \[
    (u>1\land v>0)\lor 4v>(u+1)^2,
  \]
  then Algorithm~\ref{algo:1} terminates.
\end{conjecture}

The conjecture is equivalent to saying that Algorithm~\ref{algo:1} terminates
if $x^2+ux+v$ has no positive root. 
If the conjecture is true, then about 96.35\% of the area of the triangle 
are covered by one of the two algorithms we considered.


\end{document}